\documentclass[11pt]{amsart}
\usepackage{amsfonts,amscd,amsthm,amsgen,amsmath,amssymb, hyperref}
\usepackage[all]{xy}
\usepackage[vcentermath]{youngtab}
\usepackage[letterpaper,hmargin=1.55in,vmargin=1.2in]{geometry}
\usepackage{graphicx, xcolor}
\usepackage{bm}	
\usepackage{amsmath}
\usepackage{amsthm}	
\usepackage{amsfonts}	
\usepackage{amssymb}
\usepackage{mathrsfs}
\usepackage{latexsym}
\usepackage{enumerate}
\usepackage{dsfont}
\usepackage{amssymb}

\theoremstyle{plain}
\newtheorem{theorem}{Theorem}[section]

\theoremstyle{definition}

\theoremstyle{remark}

\numberwithin{equation}{section}
\numberwithin{figure}{section}

\begin{document}


\title{Extensions of lattice groups, gerbes and chiral fermions on a torus}
\footnote{Invited talk at the
conference "String Geometries and Dualities", IMPA, Rio de Janeiro, December 12 - 16, 2016}

\author{Jouko Mickelsson}
\address{Department of Mathematics and Statistics, University of Helsinki}

\email{jouko@kth.se}

\maketitle
\begin{abstract} 
Motivated by the topological classification of hamiltonians in condensed matter
physics (topological insulators) we study the relations between chiral
Dirac operators coupled to an abelian vector potential on a torus in 3 and 1 space dimensions.
We find that a large class of these hamiltonians in three dimensions is equivalent, in K theory, to a family of hamiltonians in just one space dimension but with a different
abelian gauge group.

The moduli space of U(1) gauge connections over a torus with a fixed Chern class
is again a torus up to a homotopy. Gerbes over a n-torus can be realized in terms of extensions of
the lattice group acting in a real vector space. The extension comes from
the action of the lattice group (thought of as "large" gauge transformations,
homomorphisms from the torus to U(1)) in the Fock space of chiral fermions.
Interestingly, the K theoretic classification of Dirac operators coupled to vector
potentials in this setting in 3 dimensions can be related to families of Dirac operators on a circle with gauge group the 3-torus. 

\end{abstract}

\section{Introduction}

Topological classification of hamiltonians in condensed matter physics has 
attracted considerable interest in recent years. Families of hamilton operators
para\-metrized by some topological space $X$ lead to a study of the K group
$K(X)$ of $X.$ Typically the physical space is compactified as a torus $T^d$ 
and the hamiltonians are some simple perturbations of the (chiral) Dirac operator
on $T^d.$ Often there are discrete symmetries (time reversal, charge conjugation,
parity) which make the classification richer, in terms of the Real K theory
groups $KR(X),$ giving the 'periodic table' for topological insulators \cite{Kit},
\cite{RSFL}. 

In this paper we shall just restrict us to the complex K theory groups $K(X).$
The Dirac operators are coupled to an abelian gauge potential (Maxwell field or
its generalization).  The real interest is in the case of space dimension $d=3$
but interestingly we can show that a large class of chiral hamiltonians is actually
equivalent, in terms of K theory, to a family of hamiltonians on the circle $S^1$
with another abelian gauge group.

The canonical quantization of chiral hamilton operators leads to projective
Hilbert bundle (a gerbe) over the moduli space of gauge connections. 
In the nonabelian case these have studied earlier in \cite{CMM}.
In the abelian case there are important simplifications. It turns out that
the gerbe can be constructed from a central extension of a transformation
groupoid, the lattice group $\mathbb{Z}^3$ acting on the real vector space 
$\mathbb{R}^3.$ 

A gerbe can be alternatively described in terms of local complex line bundles
over intersections of open sets in a cover.  In the case at hand, pulling back
to $\mathbb{R}^3$ by the projection $\pi: \mathbb{R}^3 \to T^3$ the gerbe can be
trivialized with respect to the pull-back of the open cover; this leads to
a family of local complex line bundles over $\mathbb{R}^3$ with singularities at
lattice points parametrized by the $\mathbb{Z}^3$ action. Interestingly, 
these line bundles are renormalized infinite sums of monopole line bundles at the lattice points, the total curvature coming from the curvature of a certain
Grassmann manifold modelled by Schatten ideals. This construction can be viewed
as a special case of the abstract gauge group extension in \cite{MR}.

Finally, in Section 5  an alternative method for constructing gerbes over compact
Lie groups is presented in terms of finite-dimensional groupoid extensions
using a similar idea as in the construction of the central extension of the
transformation groupoid $\mathbb{R}^3 \times \mathbb{Z}^3 \to \mathbb{R}^3$ 
through canonical quantization of fermions on a torus.

In the hamiltonian quantization of chiral fermions coupled to a nonabelian
vector potential the group of gauge transformations acts through an 
abelian extension.  In particular, in the case of 3-dimensional physical
space $M$ the extension manifests as a 2-cocycle of the form \cite{Mi},\cite{FSh}
\begin{equation}
c(X,Y) = \frac{i}{24\pi^2} \int_M \text{tr}\, A (dXdY -dYdX)\label{nonab2cocycle}
\end{equation}
for a pair $X,Y$ of infinitesimal gauge transformations acting on Weyl spinors and a gauge potential
$A,$ the trace being computed in a finite-dimensional representation of a
compact gauge group $G.$ The cocycle condition here reads
$$ c([X,Y], Z) + c([Y,Z], X) + c([Z,X], Y) - L_X c(Y,Z) - L_Y c(Z,X) - L_Z c(X,Y)=0$$
with $L_X$ the Lie derivative acting on functions of the potential $A.$
Ideally, the group $\mathcal{G}$ of smooth maps
$M\to G$ would be represented in a Hilbert space arising from the quantization
of fermions and the gauge connection. However, unlike in the case of a $1+1$
dimensional space time such a (faithful) representation is not known. The
obstruction comes from the fact that no suitable measure is known to exist
in the space $\mathcal{A}$ of vector potentials in three space dimensions.

In this paper we concentrate on the case of an abelian gauge field with
gauge group $G=U(1).$ We assume that the 3-dimensional space $M$ is compactified
as the 3-torus $T^3.$ In this case the Lie algebra cocycle defining the extension 
becomes trivial, but there remains a global 2-cocycle supported in the group
of large gauge transformations $\mathbb{Z}^3=\mathrm{Hom}(T^3, U(1)).$ It turns out that 
using the subgroup $\mathcal{G}_0$ of contractible gauge transformations the
remaining degrees of freedom can be factorized as $\mathcal{A}'' \oplus \mathbb{R}^3$
where also the first summand is contractible and the large gauge transformations
$\mathbb{Z}^3$ act as translations on the space $\mathbb{R}^3$ of constant
gauge connections on the torus. 

The quantization of chiral fermions creates a 2-cocycle for the action of the group $\mathbb{Z}^3.$ 
on $\mathbb{R}^3.$ Since the residual gauge group action on the infinite dimensional
part $\mathcal{A}$ is trivial we can concentrate on the action on the fermionic
Fock spaces parametrized by the constant potentials $\mathbb{R}^3.$ Using the 
Lebesgue measure on $\mathbb{R}^3$ we can define the Hilbert space $\mathcal{H}$ 
as the space of square integrable functions on $\mathbb{R}^3$ with values in
the fermionic Fock space $\mathcal{F}.$ The extension of the gauge group
$\mathbb{Z}^3$ is then unitarily and faithfully represented in $\mathcal{H}.$

This paper was inspired by several discussions on condensed matter
problems and Schwinger terms with Edwin Langmann, which is gratefully
acknowledged.

\section{Topological classification of the Dirac operators on a 3-torus}

Let us specify the setting for 1-particle fermions.
We fix a compact 3-dimensional spin manifold $M,$ later to be fixed the 3-torus $T^3.$
Normally, the abelian extension of a gauge current algebra comes from coupling
a nonabelian vector potential $A$ (1-form on $M$ with values in the Lie algebra
$\mathfrak g$ of a compact gauge group $G$) to 2-component Weyl fermions. The canonical
quantization of the Weyl fermions induces the 2-cocycle \eqref{nonab2cocycle} on
the Lie algebra $Map(M,\mathfrak{g})$ when acting in the fermionic Fock spaces
coupled to $A.$

Let us now change the setting in the following way. Consider Dirac fermions (4 components)
coupled to the nonabelian vector potential $A$ in the usual way, and in addition let
us couple the fermions chirally to an abelian vector potential $a;$ that is, $a$ is only
coupled to the left handed fermions but not to the right handed fermions.  

The chiral anomaly is  computed using the families index theorem \cite{AS}. In four space-time
dimensions one takes the 6-form part of the index form
$$ \text{tr}\, \hat{A} e^{F/2\pi i}$$
where $\hat A$ is the A-roof genus computed from the Riemann tensor and $F$ is the curvature
form $F= dA + \frac12 [A,A].$ In most cases in 4 dimensions $\hat A$ vanishes 
(and in particular on a torus or a sphere) so let us
assume that this is the case. The 6-form part is then 
$$\frac{1}{6\cdot (2\pi i)^3} \text{tr}\, F^3.$$
The chiral anomaly and the Schwinger terms are calculated by transgression starting from
the above expression.  However, in the case of Dirac fermions the contributions from
the left and right sectors come with opposite signs and they cancel.  But now  $a$ is
chirally coupled to the fermions there is a piece which is left over, namely
$$\frac{i}{24\pi^2}  \text{tr} \, (f^3 + 3 f F^2).$$
Here $f =da$ is the field strength of the abelian
(Maxwell) potential.  Again by transgression this leads to a mixed Schwinger term
\begin{equation}
c(X,Y) =  \frac{i}{24\pi^2} \int_M  f\, \text{tr} (XdY - YdX).\label{abextension}
\end{equation}
Here $f$ is fixed as the field strength of the external Maxwell field.  Thus we get a central
extension of the current algebra $Map(M, \mathfrak{g}).$   

The above central extension has an operator theoretic derivation similar to the
construction of the current algebra in $1+1$ dimensions in the fermionic Fock space.
Now the grading operator is the sign $\epsilon_a$ of the hamiltonian $D_a = D_0 + \frac12(\gamma_5 +1) a$ where $D_0$ is the free Dirac hamiltonian. 

Then one can check that $[\epsilon_a, X]$ is conditionally Hilbert-Schmidt.
Conditionally means here that when computing traces of operators one has to take first
the trace over spin and gauge algebra indices and then perform space and momentum integration
for pseudodifferential operators.   The trace over spin indices for Dirac fermions makes that 
the diverging contributions from the left and right sector cancel. They cancel totally if
$a=0$ but in the case of $a\neq 0$ there is a left-over piece when expanding $\epsilon_a$ 
in powers of $a.$ This calculation can be done explicitly using residue calculus
\cite{Mick93}, \cite{ArnMick}. 
In the canonical quantization of fermions in $1+1$ dimensions the gauge algebra
is centrally extended and the 2-cocycle of the extension can be evaluated
(when the physical space is compactified as a circle $S^1$)
from 
\begin{equation}
\text{tr}_C  \, X[\epsilon, Y] = \mathrm{Res}\, \epsilon X[\ell, Y] = \frac{1}{2\pi i} \int_{S^1} \text{tr}\,X dY \label{1dcocycle} 
\end{equation}
for a pair $X,Y: S^1 \to \mathfrak{g}$ of infinitesimal gauge transformations where $\ell$ is a logarithmic symbol (the log of the momentum operator in one dimension)
and $\mathrm{Res}$ is the Wodzicki operator residue. Here the conditional trace $\text{tr}_C$ is 
calculated in a bases where $\epsilon$ is diagonal.
In general, the first equality is true only up to coboundaries of cocycles but it can
always be used when the Hilbert-Schmidt condition is satisfied for the off-diagonal 
blocks. The equality comes from 
$$\text{tr}_C \,X[\epsilon, Y] = - \text{TR}^Q \,[X,Y]\epsilon + \text{TR}^Q \, 
[X \epsilon,  Y]. $$
Here $\text{TR}^Q$ denotes a weighted trace in the sense of \cite{Pa}, now the
weight $Q$ is given by the Dirac operator. The first term on the right is a coboundary of a 1-cochain. The second term is a
residue by the general rule for pseudo  differential operators,

$$\text{TR}^Q \, [A, B] = - \mathrm{Res} \, A [\ell, B],$$
for $\ell = \log\, Q,$ see e.g. eq (4) in \cite{Pa}. The second equality in \eqref{1dcocycle} is then a simple consequence of asymptotic calculus of pseudo
differential symbols. 

Now let us concentrate on the abelian case with gauge group $U(1)$, $A=0$ and $X,Y: M \to i\mathbb{R}$ infinitesimal
gauge transformations acting on the abelian vector potential $a.$  We also assume that the physical space $M$ is the 3-torus $T^3.$
At  first sight the extension of the gauge algebra defined by
the cocycle $c$ is trivial since $c= \delta b$ where $b$ is the cochain
$$b(a;X) = \frac{-i}{24\pi^2} \int_{T^3}  da \wedge a X.$$
However, this is not the whole story. First,  the group $\mathcal{G}$ of gauge transformations 
$g: T^3 \to U(1)$ is disconnected,
$$\mathcal{G} = \mathbb{Z}^3 \times \mathcal{G}_0$$
where the elements of $\mathbb{Z}^3$ corresponds to the large gauge transformations
$g(x_1, x_2, x_3) = \exp{(2\pi i \bold{n} \cdot \bold{x} )}$ with $\bold{n} \in 
\mathbb{Z}^3$ and $\mathcal{G}_0$ consists of the contractible maps $T^3 \to U(1)$
which can be written as $g= \exp{(2\pi X)}$ with $X: T^3 \to i\mathbb{R}$ periodic
on the interval $[0,1]^3.$

We can define a group cocycle on $\mathcal G$ such that the corresponding  Lie algebra cocycle is $c.$ It is given by the formula
\begin{equation} C(a; g, g') = e^{2\pi i \int_{T^3}  a\wedge dX \wedge dY}\label{groupextension}
\end{equation}
with $g' = \exp{2\pi Y}.$
In particular, when restricted to $\mathbb{Z}^3$ the cocycle takes the form
\begin{equation}
C(a; \bold{n}, \bold{m}) = e^{2\pi i \bold{a} \wedge \bold{n}\wedge\bold{m}}
\label{groupextensionZ} \end{equation}
for constant potentials $\bold{a}= a_1 dx_1 + a_2 dx_2 + a_3 dx_3.$ This gives all 
topological information about the cocycle since $\mathcal{A}/\mathcal{G}_0$ is
homotopic to the space of constant vector potentials modulo the group $\pi_1(T^3)
=\mathbb{Z}^3$ of large gauge transformations, i.e., the homomorphisms from 
$T^3$ to $S^1.$ This follows from the fact that any potential $a$ on the torus
$T^3$ is gauge equivalent to a potential $A'$ with $\sum \partial_i A'_i =0$,
through a contractible gauge transformation.  The space of divergence free
potentials is a direct sum of the space of constant potentials and its
orthogonal complement, the space $\mathcal{A}''$ of those divergence free potentials $A$ 
with 
$$\int_{T^3} a_i d^3 x =0$$
for $i=1,2,3.$ The large gauge transformations act on $\mathcal{A}'' \oplus 
\mathbb{R}^3$ as translations by $\mathbb{Z}^3$ on the second summand.

Let $\mathcal{G}_b$
be the group of based gauge transformations $g$, i.e., $g(0,0,0) =1.$ This group acts
freely on $\mathcal A$ and from the above discussion it follows that
up to homotopy the moduli space
$\mathcal{A}/\mathcal{G}_b$ can be identified as the set of constant potentials, parametrized
by $\mathbb{R}^3,$ modulo the action of $\mathbb{Z}^3$, i.e., $\mathcal{A}/\mathcal{G}_b
= T^3_a.$ The moduli space torus is denoted by $T^3_a$ whereas $T^3_x$ will denote
the physical space. 

Coming back to the case of Dirac operators on the torus $T^3$ coupled to abelian
gauge potentials $a$ we again restrict to the constant potentials since the
moduli space is $\mathcal{A}/\mathcal{G}_b \simeq \mathbb{R}^3/\mathbb{Z}^3 = T^3_a.$
When $a$ is variable, the potential $a=\sum_i a_i dx_i$ is a potential on
$\mathbb{R}^3_a \times T^3_x;$ it is only locally defined on $T^3_a\times T^3_x$ but
its curvature $F= \sum_i da_i \wedge dx_i$ descends to $T^3_a\times T^3_x.$ 
With the families index formula this gives the 3-form
\begin{equation}
\Omega = \frac{1}{(2\pi i)^3} \int_{T^3_x} F^3 = da_1\wedge da_2\wedge da_3\label{DDclass} \end{equation}
on the moduli torus $T^3_a.$ 

We have assumed that the vector potential $a$ is globally defined, i.e., it
comes from a complex line bundle over $T^3_x$ with vanishing Chern class. 
In the case of non vanishing Chern class we can write a connection in the
form $\nabla + a$ where $\nabla$ is a fixed connection and $a$ is globally defined.
One can then repeat the above considerations in this case. The cohomology $\mathrm{H}^2(T^3,
\mathbb{Z})$ is isomorphic to $\mathbb{Z}^3$ so the total moduli space for all
complex line bundles over $T^3_x$ becomes $\mathbb{Z}^3 \times T^3_a.$ 

The form $\Omega$ is the Dixmier-Douady class of a complex projective vector
bundle. In the canonical quantization of fermions chirally coupled to a vector
potential the bundle of Fock spaces is defined over the covering $\mathbb{R}^3_a$
but there is an obstruction coming from a nonzero $\Omega$ to push it to a 
bundle over the moduli space $T^3_a.$ Alternatively, the obstruction is described
by the central extension of the $\mathbb{Z}^3$ action on $\mathbb{R}^3_a$ 
given by \eqref{groupextensionZ}. 

Let $S_{ijk}$ with $i,j,k=1,2,3$ be any tensor with integer entries. Then
\begin{equation}
C(a;n,m) = e^{2\pi i \sum S_{jkl} a_j n_k m_l} \label{groupextension2}\end{equation}
defines an extension for the $\mathbb{Z}^3$ action on $\mathbb{R}^3.$  
However, the Dixmier-Douady class corresponding to this extension is
$$\sum S_{ijk} da_i \wedge da_j \wedge da_k$$
and therefore depends only on the antisymmetrization $S$ which is $p \cdot \epsilon_
{ijk}$ where $p\in \mathbb{Z}$ and $\epsilon$ is the unique totally antisymmetric
tensor with $\epsilon_{123} = 1,$   \cite{MW}, Section 7. The integer $p$ is equal to
$\frac{1}{6} \sum \epsilon_{ijk} S_{ijk}.$ 

The Dixmier-Douady class is the only topological invariant of a projective
complex Hilbert bundle.  However, the families index theorem gives 
characteristic classes in any odd dimension for a family of self adjoint
Fredholm operators. In our case the parameter space is $T^3$ so the odd cohomology
is nonvanishing only in dimensions 1 and 3. The element in $\mathrm{H}^1(T^3, \mathbb{Z})$
describes the spectral flow for a family of hamiltonians.  It is again computed from the index theorem by taking the form $F^2$ on $T^3_x \times T^3_a$ and integrating
over $T^3_x.$ For a trivial complex line bundle over $T^3_x$ this gives zero
since the (3,1) component of $F^2$ vanishes. However, twisting with a non trivial 
line bundle with curvature $f= f_1 dx_2\wedge dx_3 + f_2 dx_3\wedge dx_1 + f_3 dx_1\wedge dx_2$ the (3,1) component becomes $f\wedge \sum_i da_i \wedge
dx_i$ and its integral over $T^3_x$ is equal to the 1-form $\sum_i f_i da_i$
on $T^3_a.$  

Next we show that up to homotopy the same class of projective Dirac operators
over the parameter space $T^3_a$ is obtained 1D Dirac operators on a unit circle
$S^1_{\theta}.$   One can apply the construction in \cite{HM}, Section 2. Let us write $T^3$
as the product $S^1_{\phi} \times M$ with $M=T^2.$ Fix a an element $\beta \in \mathrm{H}^2(M, \mathbb{Z}) \simeq \mathbb{Z}.$ Now $\mathrm{H}^1(S^1, \mathbb{Z}) \simeq \mathbb{Z}$
and we choose the angular form  $d\phi$ as the generator; the circle $S^1_{\phi}$
is a parameter, in addition to $M$, for a family of Dirac operators whereas
the 1D Dirac operator is defined on the circle $S^1_{\theta}.$  The parameter
$\phi$ measures the holonomy around $S^1_{\theta}$ of the Dirac operator coupled
to a constant vector potential. 

The family of Dirac operators is then twisted with a connection on $S^1_{\theta}
\times S^1_{\phi} \times M$ with total curvature $F= \frac{i}{2\pi} d\theta\wedge d\phi+ \beta$ and the index form on $S^1_{\phi} \times M$ becomes
$$ \int_{S^1_{\theta}}  e^{F/2\pi i } = \frac{d\phi}{2\pi} + \frac{d\phi}{2\pi}\wedge
\frac{\beta}{2\pi i}.$$
In particular, $\beta \mapsto d\phi \wedge \beta$ gives an isomorphism between
$\mathrm{H}^2(M, \mathbb{Z})$ and $\mathrm{H}^3(T^3, \mathbb{Z}).$ 

The above construction can be slightly generalized noting that we have arbitrarily
chosen one circle $S^1_{\phi}$ in $T^3.$ We can also take direct sums of 1D Dirac
operators corresponding to three different choices of the parameter circle $S^1_{\phi}$ inside $T^3.$ Giving weights $(f_1, f_2, f_3) \in \mathbb{Z}^3$ to the different
choices leads to the index form $\sum_i f_i dx_i$ in $\mathrm{H}^1(T^3, \mathbb{Z})$ 
and to the Dixmier-Douady class $f\wedge \beta/2\pi i$ in $\mathrm{H}^3(T^3, \mathbb{Z})$ for 
a given line bundle on $T^3$ with curvature $\beta.$ 
More concretely, we may consider a family of Dirac operators on the unit circle
$S_{\theta}$ coupled to an abelian gauge connection with the structure group
$T^3.$ Then the moduli space of gauge connections is again $T^3$ and we may
twist the family of Dirac operators by a complex line bundle over $S^1_{\theta}
\times T^3$
with curvature $f= d\theta \wedge \alpha + \beta$ where $\alpha/2\pi i$ is an integral 
1-form on $T^3$ and $\beta/2\pi i$ is an integral 2-form on $T^3.$ Then
$$\int_{S^1_{\theta}} e^{f/2\pi i} = \alpha/2\pi i + \alpha \wedge \beta/(2\pi i)^2.$$
 In conclusion, we have 

\begin{theorem}\label{1D-3D} The odd K-theory classes on $T^3$ generated by
Dirac operators $D_A$ coupled to $U(1)$ gauge connection $A$ on the torus can be
alternatively defined by 1-dimensional Dirac operators in the fibers of a circle
bundle over a 3-torus provided that the greatest common divisor of the
components $\beta_{ij}$ is equal to one. 
\end{theorem}
\begin{proof} The Chern character map from $K^*(T^3)$ to $\mathrm{H}^*(T^3, \mathbb{Z})$
is an isomorphism. The Chern character in the case of the 3D Dirac operators
described above is the generator of $\mathrm{H}^3(T^3, \mathbb{Z}^3)$ together with
the degree one component $f\in \mathrm{H}^1(T^3,\mathbb{Z}).$ In the case of the family of 1D
Dirac operators the Chern character is $f$ in degree one and $\sum \beta_i f_i$ 
times the basic form in
degree three. The latter has the value $1$ for a suitable $\beta$ if
the greatest common divisor of the components $f_i$ is one. \end{proof}
 
\section{Quantization of 1D fermions with an abelian gauge group}

As we have seen, up to homotopy, the moduli space of $U(1)^3$ gauge connections
on the unit circle $S^1$  is  a three torus $\mathbb{R}^3/\mathbb{Z}^3 = T^3_a$
where the group $\mathbb{Z}^3$ consists of the gauge transformations $(e^{i n_1 x},
e^{i n_2 x}, e^{i n_3 x})$ 
with $0\leq x \leq 2\pi$ acting on the constant potentials
$a$ as $a^i \mapsto a^i +n_i.$   Fix an element $\Omega \in \mathrm{H}^3(T^3_a, \mathbb{Z}).$ We construct a projective bundle of fermionic Fock spaces over
$T^3_a$ corresponding to the Dixmier-Douady class $\Omega.$ 

To start with consider the trivial bundle $\mathbb{R}^3 \times \mathcal{F}$
over $\mathbb{R}^3$ where $\mathcal{F}$ carries a representation of the
canonical commutation relations algebra (CAR) with a vacuum vector $|0>.$ 
The CAR algebra is generated by the elements $b^*(v)$ and $b(v)$ with 
$v\in H= L^2(S^1, \mathbb{C}^3)$ with the  nonzero anticommutators
$$b^*(u) b(v) + b(v) b^*(u) = <u,v>_{L^2}\cdot {\bf{ 1}}$$
when $u,v$ are proportional to the same basis vector $e_i \in \mathbb{C}^3.$ 
We require that
the fermions corresponding to the 3 different coordinates in $\mathbb{C}^3$
commute with each other; this is not essential, we could make them anticommute,
this is only to make certain sign conventions later on simpler.
The vacuum vector
is characterized by the property
$$ b^*(u) |0> = 0 = b(v) |0>$$
for $u\in H_-$ and $v\in H_+$ where $H_-$ (resp. $H_+$) is the subspace
spanned by the nonpositive (resp. positive) Fourier modes on the circle $S^1.$
 
Let $v_{j,p}(x) = \frac{1}{\sqrt{2\pi}} e^{ipx}e_j$ denote the Fourier modes
in the $j$th direction $e_j\in \mathbb{C}^3.$ For $n\in \mathbb{Z}^3$
set $b^*(n)= b^*(\sum_j v_{j, n_j})$ and likewise for $b(n).$ 
Next we twist these modes by a 1-cocycle over $\mathbb{R}^3$ allowing the
operators depend on $a\in\mathbb{R}^3$ such that
$$b^*(n,a+m) = e^{i a\wedge \beta\wedge m} b^*(n,a)$$
where $\beta\in\mathbb{Z}^3$ is a fixed vector.  This means that the fermion operators
are twisted by a complex line bundle over $T^3$ with curvature labelled by the
components of the vector $\beta.$ 

The action of the gauge transformations $g(m)$ in the Fock space is now completely
fixed by the condition
\begin{equation}
 g(m) b^*(n,a) g(m)^{-1} = b^*(n+m, a+m)= e^{i a\wedge \beta\wedge m} b^*(n+m,a)
\label{g-conj}
\end{equation}
and the action on the vacuum vector
\begin{equation} 
g(m) |0> = (\prod_i S_i^{\alpha_i m_i})|0>.\label{g-vacuum}
\end{equation}
Here $S_i$ is the shift operator in the $i$-direction, increasing the fermion number
by one unit;
$$S_i |0> =   b^*(v_{i,1}))|0>.$$
The vector $\alpha\in \mathbb{Z}^3$ corresponds to the spectral flow 1-form
in the previous section.
We have  $S_i N_j S^{-1}_i = N_j + \delta_{ij}$ where
$$N_j = \sum_p : b^*(v_{j,p}) b(v_{j,p}) : \,\, = \sum_{p>0} b^*(v_{j,p}) b(v_{j,p}) - \sum_{p\leq 0} b(v_{j,p}) b^*(v_{j,p}).$$ 
The shift operator $S_i$ is the quantization of the 
1-particle space shift operator $s_i v_{j,p} = v_{j,p+\delta_{j,i}}$ in the $i$th direction in $L^2(S^1,\mathbb{C}^3).$ 

Now we can compute the twisted action of $\mathbb{Z}^3$ in $\mathcal F.$ Denote by $|N>$ the subspace
with fermion numbers $N=(N_1, N_2, N_3). $ We can write
$$|N> = P_N(b^*) |0>$$
where $P_N$ is a polynomial in the creation operators $b^*$ of order $p=N_1 +N_2 +N_3.$
Using repeatedly \eqref{g-conj} we obtain
$$ g(n) |N> = e^{2\pi i p a\wedge \beta\wedge n} P_p(b^*_{+n}) g(n) |0>
= e^{2\pi i p a\wedge \beta \wedge n} P_p(b^*_{+n}) (\prod S_i^{\alpha_i n_i})|0>$$
where the subscript $+n$ indicates that all the momenta are shifted by the vector $n.$
Comparing now the  action of $g(n)g(m)$ to the action of $g(n+m)$ on $|N>$ we
get
$$g(n) g(m) = e^{2\pi i (\alpha\cdot m)) a\wedge n\beta}n g(n+m)= C(a; n,m)g(n+m).
\label{cocycle} $$
The group cohomology in degree $2$ for the $\mathbb{Z}^3$ action on the module of functions of $a\in \mathbb{R}^3$ is 1-dimensional, see the discussion in \cite{MW},
Section 7.1, specialized to the case $n=3.$ The cohomology is generated by
the cocycle $e^{2\pi i a\wedge n \wedge m}.$ One can then check by a direct computation,
projecting the cocycle to its antisymmetric form
Denote $k= \alpha \cdot \beta.$ Then the cocycle $C$ above is equal, up to a coboundary, to the cocycle 

$$C'(a; n,m)=  e^{2\pi i k a\wedge n \wedge m}$$
where $k= \alpha\cdot \beta.$ One can check the power $k$ by mapping the (antisymmetrized)
2-cocycle $C$ to the 3-cocycle  $\frac{1}{2\pi i} \delta \log C$ and integrating the
corresponding de Rham cocycle (using a David Wigners theorem as in \cite{MW}) over the 3-torus $T^3.$

Thus $\mathbb{Z}^3$ is acting through an abelian extension defined by the 2-cocycle $C$
with values in the abelian group of exponential $S^1$ valued functions on $\mathbb{R}^3.$ 
This extension is actually central since $C(a+p; n,m) = C(a; n,m)$ for all $p\in \mathbb{Z}^3.$ The integer $k$ corresponds to the Dixmier-Douady class $\Omega\in \mathrm{H}^3(T^3, \mathbb{Z}) = \mathbb{Z}.$ 

The quantized Dirac operator coupled to the constant vector potential $a$ can now
be written as
$$ \hat{D}_a = \sum_{j,p} j_p : b^*(v_{j,p}) b(v_{j,p}):  - \sum_{i=1}^{3} a_i J_0^{i} + \frac12 \sum_{i=1}^{3} a_i^2$$
where $J_0^{i}$ is the zero Fourier component of the gauge current in the $i$
direction,

$$J_0^{i} = \sum_p  : b^*(v_{i,p}) b(v_{i,p}): +\frac12  = N_i +\frac12$$  
where again the $v_{i,p}$'s are the Fourier modes in the $i$ direction in $\mathbb{C}^3.$
One can check by a direct computation that indeed the family of hamiltonians transforms covariantly under the gauge transformations,  $g(n)^{-1} \hat{D}_a g(n) =
\hat{D}_{a+n}.$

The first term in the operator $\hat D_a$ is the quantization of the free Dirac operator
$i\frac{d}{dx}$ in one dimension, i.e., the generator of rotations on a circle,
the second term is the 'minimal coupling' term, written as  $\int  : \psi(x)^*a(x) \psi(x): dx$ in the standard physics notation, now restricted to case $a=$ constant;
the last term is needed to guarantee the gauge covariance (corresponds to the last
term in \cite{Mic90}, eq. (4.5).

There is an even simpler realization of the gerbe as a projective vector bundle
(but not of the family of quantized Dirac operators) corresponding to the
same Dixmier-Douady class defined by the groupoid cocycle $c$ above.
Let $H$ be a complex Hilbert space with an orthonormal basis labelled by integers
$N.$ Fix a pair of vectors $p,q\in \mathbb{Z}^3$. Let $\mathcal{H} = L^2(\mathbb{R}^3, H)$ and for $\psi\in\mathcal{H}$ set
$$ (g(n) \psi)_N(a) = e^{2\pi iN a\wedge p \wedge n}  \psi(a-n)_{N-q\cdot n}$$
where $\psi_N(a)$ denotes the component of the vector $\psi(a)$ in the given basis.
The 2-cocycle computed from the action is
$$C'(a; n, m) = e^{-2\pi i(q\cdot m)a\wedge p\wedge n}.$$
It differs from $C$ by a coboundary for $k=p\cdot q$: The equivalence classes of $S^1$ 2-cocycles 
for the transformation groupoid $\mathbb{R}^3 \times \mathbb{Z}^3 \to \mathbb{R}^3$
correspond to equivalence classes of gerbes over $T^3$ which are classified by
$\mathrm{H}^3(T^3, \mathbb{Z}) = \mathbb{Z}.$ The antisymmetrization of $C'$ is
exactly $C$ when $k=p\cdot q.$ The Dixmier-Douady class of the gerbe is obtained from 
$\frac{1}{2\pi i} \delta \mathrm{log} C$ which generates the group cohomology $\mathrm{H}^3_{grp}(
\mathbb{Z}^3, \mathbb{Z}) = \mathbb{Z} = \mathrm{H}^3(T^3, \mathbb{Z}),$ \cite{MW}, Section 7.

-
\section{Quantization of 3D fermions with gauge group $U(1)$} 

In 3 space dimensions there is a technical problem related to the fact that
only constant gauge transformations can be canonically implemented in the
fermionic Fock space, \cite{ShSt}. However, one can circumvent this obstacle
as follows, \cite{Mick93}.  For any smooth vector potential $a$ one constructs
an unitary operator $T_a$ in the 1-particle space such that
$$ T(a,g) = T_{a^g}^{-1}g T_a$$
has the property that $[\epsilon, T(a,g)]$ is Hilbert-Schmidt for a smooth gauge
transformation $g.$ This method can be applied also in the nonabelian case.

Now $T(a,g)$ can be canonically quantized in the fermionic Fock space; the quantization
is uniquely defined up to a complex phase. In our case, restricting to the constant
vector potentials, the gauge transformation are simply shifts $a\mapsto a + n$
for $n\in \mathbb{Z}^3.$ Denote the quantized operators as $g(n).$ Now
$g(n)$ acts in the free fermionic Fock space.  But because of the nontrivial
Dixmier-Douady class in $\mathrm{H}^3(T^3,\mathbb{Z})$ computed by the
index theory argument before, the action is projective,
$$g(n) g(m) = C(a;n,m) g(n+m).$$

The structure of the gerbe over $T^3$ coming from the fermionic quantization
can be further analyzed in a very concrete manner.  Let $\lambda \in \Bbb R$
and define $U_{\lambda} = \{a\in \mathcal{A}| \lambda\notin Spec(D_a)\}.$
Then $U_0 = \{ a\in \mathbb R^3 \setminus \mathbb{Z}^3\}$ and $\mathbb{R}^3 = U_0 \cup U_{1/3}.$ The complex line bundle of fermionic vacua over $U_0$ has curvature
$$\omega =  \frac14 \text{tr}_C  (F -\epsilon) dF dF$$
where $F= D_a/|D_a|$ and $\epsilon = D_0/|D_0|$ (in the subspace of zero modes,
i.e. 3-momentum equal to zero, we fix the action of $\epsilon$  as multiplication
by zero)  and the subscript $C$ refers to the conditional trace $\text{tr}_C (X)=
\frac12 \text{tr}(X + \epsilon X\epsilon).$ In the following we write the discrete
momentum as a complex hermitean $2\times 2$ matrix $p= \sum_i p_i \sigma_i$ with
$\sigma_i$ the triple of Pauli matrices and similarly for the potential $a.$ This means that in the momentum basis
$$\omega = \frac18 \sum_{p\in \mathbb{Z}^3} \text{tr} \left(\frac{p +a}{|p+a|} - \frac{p}{|p|}\right)
d\frac{p+a}{|p+a|}\wedge d\frac{p+a}{|p+a|}$$
where again $p/|p|$ is set to zero when $p=0$ and the trace here is the $2\times 2$ matrix trace.

In a similar way the curvature of the vacuum line bundle over $U_{1/3}$ is evaluated
by replacing $D+a$ by $D+a -\frac13.$ 

In general, a gerbe is a projective Hilbert bundle over some space $X.$ It can be
described alternatively in terms of local complex line bundles $L_{\alpha\beta}$ over intersections $V_{\alpha}\cap V_{\beta}$ of elements 
$V_{\alpha}$ of an open cover of $X$ with prescribed isomorphims $L_{\alpha\beta}
\otimes L_{\beta\gamma} \equiv L_{\alpha\gamma}$ leading to the $S^1$ valued
cocycle $f_{\alpha\beta\gamma}\cdot  \bold{1} = L_{\alpha\beta}\otimes L_{\beta\gamma}
\otimes L_{\gamma\alpha}.$ Then if $\pi: Y\to X$ is a projection and $Y$ is
contractible we may write
$$\pi^*(L_{\alpha\beta}) = L_{\alpha} \otimes L_{\beta}^{-1}$$
for a family of local line bundles $L_{\alpha}\to U_{\alpha}= \pi^{-1}(V_{\alpha})$ over
$Y.$ This is exactly the case above, with $Y=\mathbb{R}^3$ and $X= T^3.$  
 
The curvature is formally a sum of two terms: The first is $\text{tr}\, F dF dF$
which is an extension of the curvature formula on a finite-dimensional Grassmannian 
to the infinite dimensional setting and the second is the exact term $\text{tr}\,
\epsilon dF dF.$ Both diverge separately when one computes the infinite sum
over momenta $p.$ However, for a fixed momentum $p$ the first term gives
$$\omega^{(1)}(p) = \frac{ \sqrt{-1}}{4}\sum_{ijk} \epsilon_{ijk}\frac{(a+p)_i\, da_j\wedge  da_k}{|p+a|^3}$$
for  $a\in\mathbb{R}^3.$ This is the curvature of a unit magnetic monopole located at
the point $-p,$  with period $2\pi\sqrt{-1}.$  So the total curvature is the sum of curvatures of magnetic
monopoles located in the infinite lattice $\mathbb{Z}^3\subset\mathbb{R}^3$
but renormalized by the subtraction of the infinite sum of exact forms located
at the same lattice points.   
 
Let us also briefly consider the case of Dirac hamiltonians in the even dimensional
case, $d=4,$ in the simple situation when the first Chern class of the $U(1)$
gauge field over the torus $T_x^4$ is zero. So again all the topological information
is in the holonomies around the four different circles in $T_x^4$ which form the
group $T_a^4.$   Now the Dirac spinors have 4 complex components and we have
the chirality operator $\Gamma$ with $\Gamma^2=1,$ anticommuting with the Dirac
operators.  

The families index theorem gives now even characteristic classes on the moduli
space $T^4_a$ of $U(1)$ gauge potentials. The cohomology in dimension 4 is
one dimensional and what corresponds to the gerbe form $\Omega$ before is
now the generator in $\mathrm{H}^4(T^4, \mathbb{Z}) = \mathbb{Z}.$ The local
trivialization on $\mathbb{R}^4$ are now local closed 3-forms. 

As before, the Dirac operator is invertible when $a\in \mathbb{R}^4 \setminus
\mathbb{Z}^4.$   On this set of potentials the 3-form
$$ \omega_3 = \text{tr}_C \Gamma(F-\epsilon) dF \wedge dF \wedge dF$$
is well-defined.  We can compute it at each momentum vector $p\in \mathbb{R}^4$
and the result is
$$ \omega^{(1)}_3(p) = \frac{(a+p) \wedge da \wedge da\wedge da}{|p+a|^4}$$
for the first term; the renormalization term involving the operator $\epsilon$
is the exact form
$$\omega_3^{(2)}(p) = d\,\text{tr} \Gamma \epsilon\frac{p+a}{|p+a|} d\frac{p+a}{|p+a|} \wedge  d\frac{p+a}{|p+a|}.$$
Summing over $p,$ both $\omega_3^{(1)}(p)$ and $\omega_3^{(2)}(p)$ diverge but
their difference is convergent. 
 
\section{Gerbes on compact simple Lie groups}

Gerbes over compact simple Lie groups have been constructed in several ways;
using the quantization of chiral fermions see \cite{CMM}, or more direct constructions
\cite{Mick2003}, \cite{Mein}, \cite{Mur}. 
As an application of the ideas in the previous sections we give another 
construction using finite dimensional groupoid extensions.

Let $G$ be a simple simply connected compact Lie group. The third cohomology
$\mathrm{H}^3(G, \mathbb{Z})$ is isomorphic to the group $\mathbb{Z}.$ Let
us fix a de Rham representative $\Omega$ of a class of level $k\in \mathbb{Z}.$ 
Let $\tilde \Omega$ be the pull-back of $\Omega$ with respect to the exponential
mapping $\exp: \mathfrak{g} \to G.$ We can then choose a 2-form $\theta$ on $\mathfrak{g}$ such that $d\theta = \tilde \Omega.$ The form $\theta$ gives in
a natural way a closed 2-form on a groupoid $EXP.$ 

The groupoid $EXP$ is defined as follows. The sources and targets of the groupoid are
point in $\mathfrak{g}.$ For a pair $s,t\in\mathfrak{g}$ there is a morphism
$s \to t$ if $e^s = e^t.$ This morphism can be realized as an element of $\Omega G$ as
$g(x) = e^{-xs} e^{xt}$ with $0\leq x \leq 1.$ This is a based gauge transformation taking
the constant vector potential $s$ on the unit circle to the constant potential $t.$

Since the morphisms are elements of the loop group we have a canonical closed 2-form
on the groupoid as the transgression of the form $\Omega$ on $G$ to the loop group.

 The set of arrows starting from a point $s$
is disconnected. For example, when $s=0$ the set $T_0$ of targets are the points $t$ such that $e^t=1.$ Restricted to the Cartan subalgebra $\mathfrak{h} \subset \mathfrak{g}$ these points are the points in the integral lattice $L$ generated by
$2\pi$ times the coroots $h_{\alpha} \in \mathfrak{h}$ and $\mathfrak{h}/L $ is the maximal torus
$H\subset G.$ The set $T_0$ consists then from the set of $G$ adjoint orbits through
$L.$ In the generic case the adjoint orbit through a point $t\in L$ is the 
smooth surface $G/H.$ 

By the homotopy exact sequence $G/H$ is simply connected, $\pi_2(G/H) = \pi_1(H) =
\mathbb{Z}^{\ell}$ with $\ell =$ rank of $G.$ It follows that also $\mathrm{H}_2(G/H,\mathbb{Z})
= \mathbb{Z}^{\ell}.$ The homology basis in $\mathrm{H}_2(G/H, \mathbb{Z})$ is given by
the 2-spheres $S^2_{\alpha}$ of $SU_{\alpha}(2)$ orbits through  $2\pi h_{\alpha}$ 
where $SU_{\alpha}(2)$ is the subgroup with Lie algebra $\mathfrak{g}_{\alpha}$ 
corresponding to the simple root vectors $e_{\pm\alpha}$ and the coroot vector $h_{\alpha}.$ 
The image of these spheres in $\mathfrak{g}$ is the unit in $G$ under the exponential mapping.

The  level $k$ on $G$ is of the form
$\Omega$ is 
$$\Omega = \frac{k}{24\pi^2} <dg g^{-1}, \frac12 [dgg^{-1}, dg g^{-1}]>$$
where the invariant form $<\cdot, \cdot >$ is normalized such that the length 
squared of the longest root becomes $2.$ The form $\theta$ can be written as
\cite{CM}
$$\theta_Z(X,Y) = \frac{k}{4\pi^2} <X, h(ad_Z) Y>$$
for tangent vectors $X,Y$ at a point $Z\in\mathfrak{g},$ with  
$$ h(z) = \frac{\sinh(z) -z}{z^2}.$$ 
On the other hand, this is obtained by integration along the paths $x\mapsto e^{xZ}$
from the 3-form $\Omega.$ When $e^Z=1$ the path is a loop and the resulting form
is just the transgression of $\Omega$ to a closed 2-form on the loop group 
restricted to loops of the form $g_Z(x) = e^{xZ}.$  In other words, up to a coboundary,
$\theta$ is the pull-back with respect to $Z\mapsto g_Z$ of the standard left invariant form

 $$\frac{k}{2\pi} \int_{S^1} <X, dY>$$
  on the loop group $\Omega G,$ 
with $X,Y$ in the loop algebra $\Omega \mathfrak{g}.$

To check the normalization of the 2-form $\theta$ we just need to pair it against
the homology cycles $S^2_{\alpha}.$ But for $Z=2\pi h_{\alpha}$ the form $\theta_Z$
gives $k/16\pi^3$ times the area form on $S^2_{\alpha}$ at the point $Z;$ 
the extra $4\pi$ in the denominator comes from $<X, h(ad_{2\pi h_{\alpha}}) Y>$
expanding $h(z)$ as a power series. On the
other hand, the radius squared of $S^2_{\alpha}$ is $||Z||^2 = ||2\pi h_{\alpha}||^2= 8\pi^2$ so the integral becomes $4\pi \cdot 8k\pi^2/ 16\pi^3= 2k.$
This is what one should expect since the ball in $\mathfrak{g}_{\alpha}$ with boundary 
$S^2_{\alpha}$ covers $SU_{\alpha}(2)$ twice in the exponential mapping.

The circle extension of $EXP$ given by the closed form $\theta$ fixes uniquely the
class of the gerbe defined by the 3-form $\Omega.$ The 3-homology of $G$ is generated
by any of the 3-spheres $SU_{\alpha}(2)$ and the integral of $\Omega$ over $SU_{\alpha}(2)$ is by Stokes theorem equal to $1/2$ times the integral of $\theta$ 
over $S^2_{\alpha}.$

\enddocument